\documentclass[12pt]{article}
\usepackage[margin=1.4in]{geometry}

\usepackage[style=ACM-Reference-Format,backend=bibtex,sorting=none]{biblatex}
\usepackage{graphicx}
\usepackage{amsthm}
\usepackage{amsmath}
\usepackage{authblk}
\usepackage{subcaption}
\usepackage[dvipsnames]{xcolor}
\usepackage{tikz}
\usetikzlibrary{arrows}

\usepackage{url}

\usepackage{breakurl}
\usepackage[pdfencoding=auto, breaklinks, psdextra, colorlinks]{hyperref}
\hypersetup{
    colorlinks = wgreen,
    citecolor = wgreen,
    urlcolor = wyellow,
    breaklinks
    }

\usepackage[ruled,vlined,linesnumbered,algonl]{algorithm2e}
\SetEndCharOfAlgoLine{}

\SetKwComment{Comment}{\footnotesize$\triangleright$\ }{}

\SetCommentSty{mycommfont}

\RequirePackage[T1]{fontenc} 
\RequirePackage[tt=false, type1=true]{libertine}
\RequirePackage[varqu]{zi4} 
\RequirePackage[libertine]{newtxmath}

\usetikzlibrary{calc}
\usetikzlibrary{shapes}
\usepackage{xcolor}

\newcommand{\wordle}{\textsf{Wordle}}
\newcommand{\np}{\mathrm{NP}}
\newcommand{\fpt}{\mathrm{FPT}}
\newcommand{\W}{\mathrm{W}[2]}

\newtheorem{theorem}{Theorem}
\newtheorem{lemma}{Lemma}
\newtheorem{corollary}{Corollary}

\definecolor{wgray}{RGB}{120,124,126}
\definecolor{wgreen}{RGB}{106,170,100}
\definecolor{wyellow}{RGB}{201,180,88}

\DeclareRobustCommand{\wordlify}[2]{
	{\protect
	\hspace{-15pt}
	\raisebox{-5pt}{
		\foreach \l [count=\c, evaluate=\c as \co using {#2[\c-1]}] in #1 {
			\ifnum \co=0
				\SquareBox{\l}{wgray}{3pt}
			\else
				\ifnum \co=1
					\SquareBox{\l}{wyellow}{3pt}%
				\else
					\SquareBox{\l}{wgreen}{3pt}%
				\fi
			\fi
			\hspace{-4pt}
		}
	}
	}
}

\usepackage[capitalise]{cleveref}
\makeatletter
\newdimen\@myBoxHeight%
\newdimen\@myBoxDepth%
\newdimen\@myBoxWidth%
\newdimen\@myBoxSize%
\DeclareRobustCommand{\SquareBox}[3]{%
	 \settoheight{\@myBoxHeight}{#3}
    \settodepth{\@myBoxDepth}{#3}
    \settowidth{\@myBoxWidth}{#3}
    \pgfmathsetlength{\@myBoxSize}{max(\@myBoxWidth,(\@myBoxHeight+\@myBoxDepth))}%
	\protect\tikz \node [shape=rectangle, color=white, text=white, minimum size=\@myBoxSize,
    			 fill=#2, inner sep=0, outer sep=0pt] 
    			 {{\fontfamily{phv}\selectfont \normalsize \textcolor{white}{#1}}};%
}%
\makeatother

\begin{document}
\title{ Wordle is \wordlify{{N,P}}{{0,2}}\hspace{-0.6em}-\wordlify{{H, A, R, D}}{{0,1,1,2}}}

\author[1]{Daniel Lokshtanov}
\author[2]{Bernardo Subercaseaux}

\affil[1]{University of California Santa Barbara}
\affil[2]{Carnegie Mellon University, Pittsburgh}
\maketitle

\abstract{Wordle is a single-player word-guessing game where the goal is to discover a secret word $w$ that has been chosen from a dictionary $D$. In order to discover $w$, the player can make at most $\ell$ guesses, which must also be words from $D$, all words in $D$ having the same length $k$. After each guess, the player is notified of the positions in which their guess matches the secret word, as well as letters in the guess that appear in the secret word in a different position. We study the game of Wordle from a complexity perspective, proving NP-hardness of its natural formalization: to decide given a dictionary $D$ and an integer $\ell$ if the player can guarantee to discover the secret word within $\ell$ guesses. Moreover, we prove that hardness holds even over instances where words have length $k = 5$, and that even in this case it is NP-hard to approximate the minimum number of guesses required to guarantee discovering the secret word (beyond a certain constant). We also present results regarding its parameterized complexity and offer some related open problems.}

\section{ \texorpdfstring{\wordlify{{I,N,T,R,O}}{{0,2,1,2,1}}}{Intro} }

\emph{Wordle} ({\scriptsize \href{https://www.nytimes.com/games/wordle/index.html}{https://www.nytimes.com/games/wordle/index.html}}) is a single-player web-based word-guessing game that combines principles of classic games such as Mastermind, Hangman, and Jotto, which have been studied from a computational perspective before~\cite{barbay2020computational, stuckman2005mastermind, 10.1007/978-3-642-30347-0_36, 2009, Ganzfried2012}.
Created by Josh Wardle for his partner~\cite{wordle3}, and published in October 2021, Wordle has reached an unprecedented level of virality and gained millions of daily players~\cite{wordle2, wordle1}.  As for the popularity of the game, Wardle suggested that part of the game's success is due to its artificial scarcity, as it can only be played once a day~\cite{wordle4}. Others have pointed at the distinctive colorful patterns with which players share their results~\cite{wordle5}. This article can be seen as yet another reason for Wordle's success: its intrinsic complexity. The relationship between in games and their computational complexity has been proposed by multiple authors~\cite{Eppstein, 10.5555/1541932}, and thus by establishing theoretical results of hardness for Wordle we can explain part of its challenging nature, and make the reader feel better about their unsuccessful guesses.

The game is played over a finite set of words $D$ (usually called \emph{dictionary}), all of which have length $k$, that is assumed to be fully known to the player. Formally, for a given alphabet $\Sigma$, $D \subseteq \Sigma^k$. The player, who we will denote the \emph{guesser}, wishes to discover a secret word $w \in D$ through a series of at most $\ell$ guesses $p_1, p_2, \ldots, p_\ell$, all of which must be words belonging to $D$.  The guesser is said to win if one of its guesses $p_i$ is exactly equal to $w$, and she loses if no such match occurs after $\ell$ guesses. Whenever a guess $p$ is made, the guesser receives some information about the relation between $p$ and the secret word $w$: she is notified of every position $i$ such that $p[i] = w[i]$, and also of positions $i$ such that the letter $p[i]$ is present in the word $w$ but not in position $i$. More precisely, to handle the case of repeated letters, we will use a notion of \emph{marking} with colors. Iterating over indices $i$, from $1$ up to $k$, we say  $w[i]$ marks $p[i]$ with green iff $w[i] = p[i]$. If $w[i] \neq p[i]$, and the set \[
S_i = \{ j : p[j] = w[i], \; p[j] \text{ is unmarked}\}
\] is not empty, then $w[i]$ marks $p[\min(S_i)]$ with yellow. All letters $p[j]$ that were not marked after iterating over the indices are then marked with gray.
 A few games of Wordle are illustrated in \Cref{fig:wordle-games}.

\begin{figure}
\centering
	\begin{tikzpicture}
		\node at (0,1) {$w$ = A\,B\,B\,E\,Y};
		\node at (0, 0) {$p_1$ \; \wordlify{{A,L,G,A,E}}{{2,0,0,0,1}}};
		\node at (0, -.60) {$p_2$ \; \wordlify{{K,E,E,P,S}}{{0,1,0,0,0}}};
		\node at (0, -.60*2) {$p_3$ \; \wordlify{{O,R,B,I,T}}{{0,0,2,0,0}}};
		\node at (0, -.60*3) {$p_4$ \; \wordlify{{B,R,I,B,E}}{{1,0,0,1,1}}};
		\node at (0, -.60*4) {$p_5$ \; \wordlify{{A,B,B,O,T}}{{2,2,2,0,0}}};
		\node at (0, -.60*5) {$p_6$ \; \wordlify{{A,B,B,E,Y}}{{2,2,2,2,2}}};
		
		\node at (4,1) {$w$ = K\,E\,B\,A\,B};
		\node at (4, 0) {$p_1$ \; \wordlify{{A,B,B,E,Y}}{{1,1,2,1,0}}};
		\node at (4, -.60) {$p_2$ \; \wordlify{{B,A,B,E,S}}{{1,1,2,1,0}}};
		\node at (4, -.60*2) {$p_3$ \; \wordlify{{K,E,E,P,S}}{{2,2,0,0,0}}};
		\node at (4, -.60*3) {$p_4$ \; \wordlify{{K,E,B,A,B}}{{2,2,2,2,2}}};	
		
		\node at (8,1) {$w$ = H\,I\,P\,P\,Y};
		\node at (8, 0) {$p_1$ \; \wordlify{{C,R,A,N,E}}{{0,0,0,0,0}}};
		\node at (8, -.60) {$p_2$ \; \wordlify{{B,O,I,L,S}}{{0,0,1,0,0}}};
		\node at (8, -.60*2) {$p_3$ \; \wordlify{{G,U,M,M,Y}}{{0,0,0,0,2}}};
		\node at (8, -.60*3) {$p_4$ \; \wordlify{{K,I,D,D,Y}}{{0,2,0,0,2}}};
		\node at (8, -.60*4) {$p_5$ \; \wordlify{{J,I,F,F,Y}}{{0,2,0,0,2}}};
		\node at (8, -.60*5) {$p_6$ \; \wordlify{{F,I,Z,Z,Y}}{{0,2,0,0,2}}};

		\end{tikzpicture}
	\caption{Illustration of three different instances of Wordle. In all of them $k = 5$ and $\ell = 6$. The first two games are won by the guesser, while the third one is lost.}
	\label{fig:wordle-games}
\end{figure}
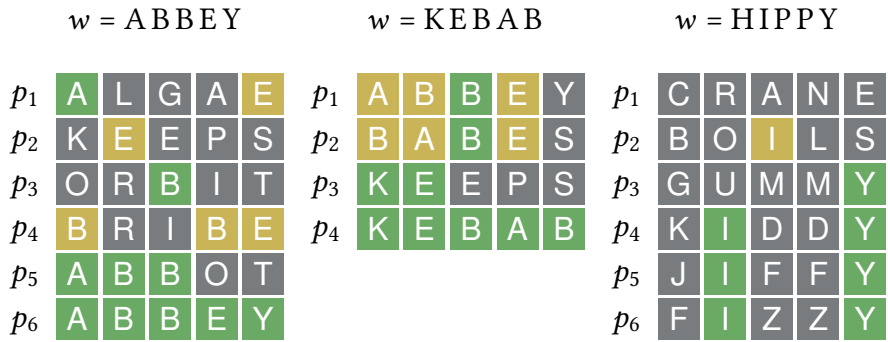

In the actual online version of Wordle, the secret word is always chosen from a fixed subset of the dictionary consisting of the most \emph{common} English words, while the guesses can be part of a larger fixed English dictionary. We do not consider this distinction in our paper as our results of hardness hold already in the case where the same dictionary is used for the secret word and the valid guesses. 

Despite its apparent simplicity, Wordle allows for encoding hard problems, as we present in the next section.

\section{Complexity \texorpdfstring{\wordlify{{R,E,S,U,L, T, S}}{{0,0,1,2,0,0,2}}}{Results}}

In order to study the complexity of Wordle, we need to define an appropriate formal decision problem for it. Let us assume that the guesser wants to design a strategy that ensures winning the game (i.e., guessing correctly within $\ell$ attempts) regardless of what the secret word was chosen to be. Thus, we can pose the following problem:

\begin{center}
\fbox{\begin{tabular}{ll}
\small{PROBLEM} : & {\wordle}
\\{\small INPUT} : & $(D, \ell)$ with $D \subseteq \Sigma^k$ the dictionary,  \\
& and an integer $\ell \geq 1$, the number of guesses allowed. 
\\ 
{\small OUTPUT} : & \textbf{Yes}, if there is a winning strategy for the guesser,  \\
& and \textbf{No} otherwise.
\\
\end{tabular}}
\end{center}

Note that $\Sigma$ and $k$ are not explicitly part of the input, but can trivially be deduced from $D$.
We claim now that if a guesser knew a polynomial time algorithm $A$ for this problem, then they could use $A$ to win the game whenever it was possible, as we explain next. 
First, we will define the notion of feasible words. Observe that at any point in the game, the information the guesser has received thus far (i.e., the letters that have been marked or not in all the previous guesses) defines a dictionary $D' \subseteq D$ of words that  could happen to be the secret word, which we call \emph{feasible} words. For example, after the first guess illustrated in the left-most game of \Cref{fig:wordle-games}, the word \texttt{GAMES} would be infeasible with the information received for the first guess, as the letter \texttt{G} was not marked. Similarly, the word \texttt{KEEPS} used in the second guess of that same game is also infeasible with the result of the first guess, as it does not contain the letter \texttt{A} that was marked in the first position. The word \texttt{AMAZE} would not be feasible either, as in the first guess the letter \texttt{E} at the end is marked with yellow, and thus cannot appear in that position in the secret word.  In contrast, the word \texttt{ANNEX} would be feasible, and naturally \texttt{ABBEY}, the actual secret word of that game, is feasible too. 
Formally, a word $p \in D$ is feasible after a sequence of guesses $p_1, \ldots, p_n$ if and only if the following conditions hold:
\begin{enumerate}
	\item No letter $p[i]$ was marked gray in a previous guess $p_j$, for $1 \leq i \leq k, 1 \leq j \leq n$.
	\item  No letter $p[i]$ was marked yellow in a previous guess $p_j$, such that $p_j[i] = p[i]$, for $1 \leq i \leq k, 1 \leq j \leq n$.
	\item If some previous guess $p_j$ was marked green in $p_j[i]$, then $p[i] = p_j[i]$. This must hold for $1 \leq i \leq k, 1 \leq j \leq n$.
	\item If some previous guess $p_j$ was marked yellow in $p_j[i]$ then $p[t] = p_j[i]$ for some $1 \leq t \neq i \leq k$.
\end{enumerate}

\begin{algorithm}
	\caption{WordleGuesser$(D, \ell)$}
	\label{alg:wordle}
	\If{$\ell = 0$}{
		guesser loses.
	}
	\If{$|D| = 1$}{
		guess the only word in $D$ and win.
	}
	\For{$p \in D$}{
		potentialGuess $\gets$ \textbf{true}\;
		
		\For{$w \in D$} {
			$ m \gets$  marking for guess $p$ if $w$   is the secret word\;
			$D' \gets  \{ p' \in D \mid p' \text{ is feasible after } m\}$\;
			\If{$WordleGuesser(D', \ell-1)$ is a loss for the guesser}{
				potentialGuess $\gets$ \textbf{false}\;
				\textbf{break}
			}
			
		}
		\If{potentialGuess} {
			Guesser can win using $p$ as its next guess.\;
		}
	}
	If this point is reached, then the guesser loses.\;
	
\end{algorithm}

Now, provided a guesser can solve the \wordle\ problem efficiently, they can play optimally by following~Algorithm~\ref{alg:wordle}. To see correctness, we can proceed by induction over $\ell$. The base case $\ell = 0$ is trivially correct. For the inductive case, if $|D| = 1$ then also correctness is trivial. In any other case, if there is a winning strategy for the guesser within $\ell$ attempts, that implies that after the first guess, regardless of what the secret word is, it will be possible to win within $\ell-1$ attempts over the dictionary restricted to the corresponding feasible words. Thus, the inductive hypothesis guarantees that calls to the algorithm with parameter $\ell - 1$ will be correct, thus implying correctness of the algorithm.
This justifies the choice of \wordle\ as a decision problem. Our next step is to study its complexity. We will prove the following theorems:

\begin{theorem}
\wordle\ is $\np$-hard, and it is $\W$-hard when parameterized by $\ell$, the number of guesses allowed.
\label{thm:1}
\end{theorem}

\begin{theorem}
\wordle\ cannot be solved in polynomial time unless $\mathrm{P} = \np$ even when restricted to instances where $k=5$.
\label{thm:2}
\end{theorem}

\begin{theorem}
	\wordle\ can be solved in polynomial time if the alphabet size $\sigma = |\Sigma|$ is constant.
\label{thm:3}
\end{theorem}

\newcommand{\F}{\mathcal{F}}
In order to prove \cref{thm:1}, we will reduce from \textsf{Almost Set Cover}. An input of \textsf{Almost Set Cover} is a pair $(\mathcal{F}, c)$, where $\mathcal{F}$ is a family of sets $S_1, \ldots, S_{|\mathcal{F}|}$ whose union we denote by $U$, and $c \geq 1$ is an integer. The goal is to decide whether there is a sub-family $\mathcal{F}' \subseteq \mathcal{F}$ of at most $c$ sets such that \[
\left| U \setminus \left(\bigcup_{S \in \mathcal{F}'} S\right) \right| \leq 1.
\]
In other words, the goal is to decide whether it is possible to cover all the elements of the universe, except for perhaps one of them, with $c$ sets. \textsf{Almost Set Cover} is  $\W$-hard (parameterized by $c$), as we show next. It is a simple reduction from standard \textsf{Set Cover}, which is known to be $\W$-hard.

\begin{lemma}
	\textsf{Almost Set Cover} is $\W$-hard, when parameterized by $c$.
	\label{lem:1}
\end{lemma}
\begin{proof}
Let $(\F, c)$ be an instance of \textsf{Set Cover}, and let $U = \bigcup \F$ be its universe. Create first $U'$ of size $2|U|$ in the following way: for each element $u \in U$ put elements $u^+$ and $u^-$ in $U'$. Now, create $\F'$ with $|\F|$ sets as follows: for each set $S \in \F$, put into $\F'$ a set
\[
 S' = \{ u^+, u^- \mid u \in S \}.
\]
We claim $(\F, c) \in \textsf{Set Cover}$ iff $(\F', c) \in \textsf{Almost Set Cover}$. For the forward direction, if $\F^* \subseteq \F$ is a set cover of size $c$ for $\F$, then its corresponding sub-family
\[
	\F^{*'}  = \{ S' \mid S \in \F^*\}
\] 
is trivially a set cover for $\F'$. On the other hand, if no sub-family $K$ of at most $c$ sets in $\F$ covers $\F$, that means that for every such sub-family $K$, there is at least one element $u \not \in \bigcup K$. But now if we consider 
\[
	K' = \{ S' \mid S \in K\},
\] 
then both $u^+$ and $u^-$ are not covered by $K'$. This implies that every sub-family $K' \subseteq \F'$ of size at most $c$ leaves at least two elements outside. This concludes the reduction, as it can be clearly computed in $\fpt$-time.
\end{proof}

We are now ready to prove \cref{thm:1}.
\begin{proof}[Proof of Theorem~\ref{thm:1}]
Let $(\mathcal{F}, c)$ be an instance of \textsf{Almost Set Cover}.
We will build a dictionary $D_\F$ as follows. First, we will need $|U|+2$ symbols, that we will denote $\Sigma_\F = \{\bot, 1, s_1, \ldots, s_{|U|}\}$. Now,  $D_\F$ will contain two different kinds of words, \emph{element-words} and \emph{set-words}, both of which will have length $|U|$. For every element $u_i \in U$, build its corresponding word $w_{u_i}$ as:
\[
	w_{u_i}[j] = \begin{cases}
		1 & \text{if } i = j\\
		s_i & \text{otherwise.}
	\end{cases}
\]
All words $w_{u_i}$ created this way constitute the set of element-words. Now, for every set $S_i \in \F$, build its corresponding set-word $w_{S_i}$ as:
\[
	w_{S_i}[j] = 
	\begin{cases}
		1 & \text{if } u_j \in S_i\\
		\bot & \text{otherwise.}
	\end{cases}
\]
We now claim that $(\mathcal{F}, c) \in \textsf{Almost Set Cover}$ iff $(D_\F, c+1) \in \wordle$. 	For the forward direction, if there is a family of sets $\F'$, with $|\F'| \leq c$ and such that $|U \setminus \bigcup \F'| \leq 1$, then the guesser can use the set-words $w_S$ for every $S \in \F'$ to guarantee a win. 
Indeed, fix some ordering $S_1, \ldots, S_{|\F'|}$ of $\F'$ and start with guess $w_{S_1}$. If the secret word $w$ was chosen to be a set-word, then two cases can arise after the first guess $w_{S_1}$: either the secret word $w$ is revealed to be exactly $w_{S_1}$, in which case the guesser wins on her first turn, or it is another set-word $w_{S'}$, in which case a symbol $\bot$ will be marked yellow, and we claim that  $w_{S'}$ will be the only feasible word remaining and thus the guesser will win in her second turn. Indeed, if we consider any element $u_i \in U$, there are 4 cases:

\begin{enumerate}
	\item $(u_i \in S, u_i \in S')$. In this case $w_{S}[i] = w_{S'} = 1$, and it will be marked green. This makes all words that do not have a $1$ in their $i$-th position infeasible.
	\item $(u_i \in S, u_i \not \in S')$. In this case $w_{S}[i] = 1$, and it will not be marked green. This makes all words that have a $1$ in their $i$-th position infeasible.
	\item $(u_i \not\in S, u_i \in S')$. In this case $w_{S}[i] = \bot$, and it will not be marked green. This makes all words that have a $\bot$ in their $i$-th position infeasible.
	\item $(u_i \not\in S, u_i \not\in S')$. In this case $w_{S}[i] = w_{S'} = \bot$, and it will be marked green. This makes all words that do not have a $\bot$ in their $i$-th position infeasible.
\end{enumerate}
Note that regardless of the case, words that do not agree with $w_{S'}$ on their $i$-th position become infeasible, and as this happens for every $i \in [U]$, the only remaining feasible word is $w_{S'}$. As $c+1 \geq 2$, and we have shown that if the secret word is a set-word then the guesser can win in at most 2 turns, we can safely proceed to the case where the secret word is an element-word. For this case the guesser can first make $c$ guesses corresponding to all the words $w_S$ for $S \in \F'$.  Let $w_{u_i}$ be the secret element word. As $|U \setminus \bigcup \F'| \leq 1$, there are two cases: in the first case, there is some $S_j \in \F'$ such that $u_i \in S_j$. Whenever the guess $w_{S_j}$ is made, as it will happen that $w_{S_j}[i] = w_{u_i}[i] = 1$. This will reveal a $1$ in the $i$-th position of the secret word, and as it is an element-word, this unambiguously reveals $w_{u_i}$ to the guesser, who will win in the next turn, and thus in no more than $c+1$ guesses. The second case is when $w_{u_i}$, the secret word, is the only element of $U \setminus \bigcup \F'$, and thus the guesser simply says that word in its next turn, also ensuring a win.
 This concludes the forward direction.
 
 For the backward direction, assume every sub-family $\F' \subseteq \F$ leaves at least two elements $u_1$ and $u_2$ of $U$ uncovered. Then we claim that no strategy for the guesser can guarantee a win within $c+1$ attempts.  Indeed, for any set of $c$ initial guesses the guesser can make, we claim that there are at least two feasible words remaining, and thus the guesser cannot ensure a win in her last turn. To see this, consider any sequence $W = w_1, \ldots, w_c$ of $c$ guesses, and then define 
 \[
 \F_\text{sets} = \{ S \in \F \mid w_S \in W\} \quad ; \quad \F_\text{elem} = \{ \{u\} \subseteq U \mid w_u \in W\}.
 \]
 Now observe that every element-word $w_u$ such that $u \not \in \bigcup{\left(\F_\text{sets} \cup \F_\text{elem}\right)}$, is feasible after the sequence of guesses $W$. Also notice that for every $S \in \F_\text{elem}$ there is a set $S' \in \F$ such that $S \subseteq S'$. Thus, if for every set $S \in \F_\text{elem}$ we choose a set  $S' \in \F$ such that $S \subseteq S'$, we get a collection of sets $\F'$ such that $\left( \bigcup \F_\text{elem}  \right)\subseteq  \left( \bigcup \F' \right)$ and $|\F'| \leq |\F_\text{elem}|$. This implies that
 \(
 	\left(\F_\text{sets} \cup \F'\right) \subseteq \F \)
and $	\left(\F_\text{sets} \cup \F'\right) $  is a collection of at most $c$ sets, and thus there are at least two elements $u_1, u_2 \in U$ that are not on its union. We thus also have that 
 \[
\{u_1, u_2\} \cap \left( \bigcup{\F_\text{sets} \cup \F_\text{elem}}\right) = \emptyset,
 \]
 and consequently, both $w_{u_1}$ and $w_{u_2}$ are feasible words, which implies it is not possible to guarantee a win in the next guess. This concludes the proof.
 
\end{proof}

\Cref{thm:2} captures an arguably essential part of Wardle's Wordle: difficulty does not require long words. Its proof is inspired by that of coNP-hardness for \textsf{Evil Hangman} by Barbay and Subercaseaux \cite{barbay2020computational}. We will use a result in hardness of approximation to prove that {\sf Wordle} cannot be solved efficiently unless $\mathrm{P} = \np$. In particular, we will use that $\gamma(G)$, the size of the smallest dominating set in a $4$-regular graph $G$, is $\np$-hard to approximate within $1+\frac{1}{390}$~\cite{Chlebk2006, Chlebk2008}. For our purpose it will be enough to use that is $\np$-hard to compute a $(1+\epsilon)$-approximation. 

\begin{proof}[Proof of Theorem~\ref{thm:2}]

Let $G$ be a 4-regular graph. We will work with words of length  $k = 5$. The alphabet $\Sigma$ will be the vertex set of $G$. We build a dictionary $D_G$ from $G$ as follows: for every vertex $v \in G$, we create two words, $w_v$ and $w'_v$: the word $w_v$ has $v$ as its first symbol is $v$, and its 4 remaining symbols are the neighbors of $v$ (the order of these 4 symbols can be chosen arbitrarily), while the word $w'_v$ consists simply of the symbol $v$ repeated $5$ times.
Given a dictionary $D$, let $W(D)$ be the minimum value of $\ell$ such that $(D, \ell)$ is a positive instance of {\sf Wordle}. We now claim that $\gamma(G) \in [W(D_G)-4, W(D_G)]$. Note that this claim is enough to prove the theorem, as a polynomial time algorithm for \textsf{Wordle} would imply that $W(D_G)$ can be computed in polynomial time, and thus $\gamma(G)$ could be approximated up to an additive constant, contradicting its $1+\frac{1}{390}$ hardness of approximation~\cite{Chlebk2006, Chlebk2008}.
	In order to prove $\gamma(G) \geq W(D_G)-4$, we show that a dominating set of size $d$ for $G$ implies a guessing strategy that guarantees a win within $d+4$ guesses.  Assume $G$ has a dominating set $S$ of size at most $d$ and the secret word is either $w_u$  or $w'_u$ for some $u \in V(G)$. Now consider the sequence of guesses $w_v, v \in S$. If $u \in S$ we have two cases, either the secret word was $w_u$, and thus the secret word was already guessed in $d$ guesses, or the secret word was $w'_u$, in which case the first symbol of guess $w_u$ was marked green, leaving only $w'_u$ as a feasible word, and thus allowing a guaranteed win within $d+1$ guesses.
	If $u \not \in S$, then as $S$ is a dominating set $u$ must be a neighbor of some vertex $v^\star$ in $S$. If the secret word was $w_u$, then the symbol $u$ in guess $w_{v^\star}$ must have been marked yellow, which allows to find $w_u$ by guessing the sequence of words $w_{u'}$ for $u' \in N(v^*)$, which must contain $w_u$, as $u$ is a neighbor of $v$. This guessing strategy has at most $d+4$ guesses and it is guaranteed to succeed. 
	On the other hand, if the secret word was $w'_u$, then the symbol $u$ in guess $w_{v^\star}$ must have been marked green, which makes $w'_u$ the only feasible word, and thus allows to win in $d+1$ guesses. 
	In order to prove $\gamma(G) \leq W(D_G)$, we show that is not possible to guarantee a win within $\gamma(G)-1$ guesses. Indeed, any sequence $\sigma$ of $\gamma(G)-1$ guesses induces a set $S$ of vertices by considering the first symbol of each guess. As $|S| \leq \gamma(G)-1$, there needs to be a vertex $v$ that is not dominated by $S$,  the word $w'_v$ is still a feasible word that is not part of $\sigma$, and thus in the case where every request in $\sigma$ was entirely marked with gray, then $\sigma$ does not make the guesser win. As this is true for any sequence $\sigma$ of length $\gamma(G) -1$, we conclude $W(D_G) \geq \gamma(G)$. 

%
%
%
	\end{proof}
We naturally obtain hardness of approximation as a corollary.
\begin{corollary}
	Let $W(D)$ be the minimum number of guesses $\ell$ such that the guesser can guarantee to win a game of Wordle over dictionary $D$ within $\ell$ guesses, i.e., $(D, \ell) \in {\sf Wordle}$. Then it is $\np$-hard to approximate $W(D)$ within $1 + \frac{1}{390}$.
\end{corollary}

Interestingly, both the proof of \cref{thm:1} and \cref{thm:2} crucially depend on variable length alphabets. This if further confirmed by \cref{thm:3}, which will we prove next. First, consider the following lemma:
\begin{lemma}
$\wordle$ restricted to instances where $\ell \leq c$, for some fixed constant~$c$, can be solved in polynomial time.
\label{lemma:2}
\end{lemma}
\begin{proof}
This can be seen by analizing the branching factor and depth of the associated game-tree. The game-tree has a branching factor of $O(D)$, as at any point the guesser can choose one of the $O(D)$ feasible words remaining, and for each such word there are at most $O(D)$ possible responses the guesser can get, depending on what the secret word $w$ was at that point. The game-tree has also depth at most $\ell$, which implies thus that the size of the game-tree is at most $O(D^\ell) = D^{O(1)}$, and thus the optimal guess at any given point can be computed in polynomial time. The result follows from this.
\end{proof}

Now we can a prove a key lemma relating $\sigma = |\Sigma|$ and $\ell$.

\begin{lemma}
On a game of \wordle\ over an alphabet $\Sigma$ of size $\sigma = |\Sigma|$, the guesser can always win in at most $\sigma$ guesses.
\label{lemma:3}
\end{lemma}

\begin{proof}
	Consider the following simple algorithm for the guesser: \emph{at any point of the game, choose any feasible word as a guess}. We will show that this algorithm requires at most $\sigma$ guesses to find the secret word. Indeed, for each index $1 \leq i \leq k$, let us define sets $S(i) \subseteq \Sigma$ as:
	\[
		S(i) = \{ s \in \Sigma \mid \text{there is a feasible word } w \in D, \text{ such that } w[i] = s \}.
	\]
	Note that, before the first guess, $S(i) = \Sigma$ for every $i$. The key idea is now the following: after a feasible word $w$ is guessed by the aforementioned strategy, the following holds for every $1 \leq i \leq k$: $w[i]$ can either be marked green, in which case $S(i) \leftarrow \{ w[i] \}$, or it can be marked with gray or yellow, in which case $S(i) \leftarrow S(i) \setminus \{ w[i] \}$. This implies that after every feasible guess each set $S(i)$ either becomes a singleton or decreases its size by $1$. Therefore, by doing $\sigma - 1$ feasible guesses, either the secret word has been found (in which case we are done), or it happens that every set $S(i)$ must be a singleton, and thus there will be only one remaining feasible word, which allows the guesser to win within $\sigma$ guesses in total.
\end{proof}

We are now ready to use the previous lemmas and prove \cref{thm:3}.

\begin{proof}[Proof of \cref{thm:3}]
	Consider an instance $(D, \ell)$ of \wordle\ where $|\Sigma| \in O(1)$. By \cref{lemma:3}, if $|\Sigma| \leq \ell$, then simply answer \textbf{Yes}. Otherwise $\ell \leq |\Sigma| = O(1)$, and thus by \cref{lemma:2} the result follows.
\end{proof}

\section{\texorpdfstring{\wordlify{{O, P, E, N}}{{0,0,1,0}}}{Open} Problems}

We now offer  problems that our analysis leaves open.

\begin{enumerate}
	\item Is \wordle\ in $\np$, or is it $\mathrm{PSPACE}$-complete? It is not hard to see that \wordle\ is in $\mathrm{PSPACE}$, by simply computing its associated game-tree, which has polynomial depth by~\cref{lemma:3}. \textbf{Latest Update:} Rosenbaum has published a proof of membership in $\np$~\cite{https://doi.org/10.48550/arxiv.2204.04104}, and it extends as well to cover the third problem on this list. 
	\item Is \wordle\ in $\fpt$ when parameterized by $\sigma = |\Sigma|$? The proof of \cref{thm:3} provides an algorithm in time $D^{O(\ell)}$, and it seems challenging to obtain a $\textrm{poly}(D) \cdot f(\ell)$ algorithm for some computable function $\ell$.
	\item How  is the complexity of \wordle\ affected by having a dictionary that is not presented as a list of word, but rather in some more compact representation as a finite automaton? We can envision $D$ to be provided as a finite automaton, for which all accepted words of length $k$ are considered part of the dictionary. This affects for example the proof of \cref{thm:3}, as now a branching factor of $|D|$ can be exponential in the input size.
\end{enumerate}

\section*{Acknowledgements}
	The second author thanks J\'er\'emy Barbay for his insightful conversation, and his awesome friends: Saranya Vijayakumar for proofreading the English in an earlier draft of this paper, and Bailey Miller for providing him a place (and tea) to work on this article. We also thank anonymous reviewers for helpful suggestions, Mark Polyakov for pointing out a minor bug in a previous proof of Theorem 1, and Wassim Bouaziz for pointing out some imprecisions in a previous version of this paper. 
	
\section*{Funding statement}
Lokshtanov is supported by BSF award 2018302 and NSF award CCF-2008838. Subercaseaux is supported by NSF
awards CCF-2015445, CCF1955785, and CCF-2006953.

\printbibliography
\end{document}